\documentclass[
11pt]{amsart}
\usepackage{
amssymb}

\numberwithin{equation}{section}

\newtheorem{Thm}{Theorem}[section]

\newtheorem{Cor}[Thm]{Corollary}
\newtheorem{Lem}[Thm]{Lemma}
\newtheorem{Prop}[Thm]{Proposition}

\newtheorem{Rmk}{Remark}

\newtheorem{Exam}[Thm]{Example}

\newcommand{\N}{\mathbb{N}}

\newcommand{\R}{\mathbb{R}}

\renewcommand{\div}{{\rm div}}

\begin{document}

\date{\today}

\title[Limits at infinitely many molecules]{
On Hydrodynamic equations\\ at the limit of 
infinitely many molecules}

\author{S. Dostoglou}
\author{N.C. Jacob}
\author{Jianfei Xue}

\address{Department of Mathematics, 
          University of Missouri, 
          Columbia, MO 65211}

\begin{abstract}
We show that weak convergence of point measures and $(2+\epsilon)$-moment conditions imply hydrodynamic equations at the limit of infinitely many interacting molecules.
The conditions are satisfied whenever the solutions of the classical equations for $N$ interacting molecules obey uniform in $N$ bounds. As an example, we show that this holds when the initial conditions are bounded
and that the molecule interaction, a certain $N$-rescaling of potentials  that include all $r^{-p}$ for $1<p$, is weak enough at the initial time. 
In this case the hydrodynamic equations coincide with the macroscopic equations of Maxwell. 
\end{abstract}

         
\maketitle

\section{Introduction}

Derivations of macroscopic hydrodynamic equations from microscopic dynamics go back to the introduction of probabilistic methods in the description of molecular motions in \cite{Max} and continue in now classic works, \cite{B}, \cite{CC}, \cite{IK}, \cite{Gr}, \cite{M}, \cite{L}, to this day, see \cite{EP} and \cite{GK} for recent reviews and references.

In the present article we examine macroscopic hydrodynamic equations as limits of the classical equations of motion for a system of $N$ interacting molecules, as $N$ becomes infinite. In particular, we subscribe to the idea that, whereas a classical system can be fully described by these equations for a finite but extremely high $N$, a reasonable approximation of an observer's macroscopic perception is the limit at infinite $N$. Our motivation has been to substantiate Reynolds's tenuous definition of hydrodynamic averages in \cite{R} and his claim that averages, as in \cite{Max} and \cite{R}, can only be space averages. 
For this, we use the classical equations of molecular motion and we average in space (or time-space) but we do not use the Liouville equation and we do not average in phase space (Gibbs ensembles). At the same time, we make no assumptions of binary collisions, molecular chaos etc. and therefore we do not use the Boltzmann equation. In this sense our work has origins in the first part of Morrey \cite{M} and Jepsen \& ter Haar \cite{JtH}.

Our starting point is the article \cite{D} where weak convergence of empirical position-velocity probability measures on $\R^6$ and disintegration of the limit measure $M$ with respect to its marginal $\mu$, the macroscopic density, provide a rigorous definition of a macroscopic velocity $u$ as the barycentric projection of the disintegration (formulas (3.14) and (5.27), loc. cit.). The tools there are from \cite{AGS}. 
The kinetic energy of $u$ is, in general, only part of the total kinetic energy of the macroscopic system (formulas (5.11) and (5.12) in \cite{D}).
This allows for part the remaining total kinetic energy at the limit to include heat and, possibly, other fluctuations. 
Following Morrey \cite{M}, it is assumed in \cite{D} that the total mass, energy, and moment of inertia stay bounded in $N$ and, to deal with the non-linear terms,  that second moments locally converge.  
With these, \cite{D} shows how the limits of equations for $N$ molecules, rescaled by a factor $\sigma_N$ at each $N$ as in \cite{M}, can give at the limit $N \to \infty$ weak versions of macroscopic equations for the limit molecule density $\mu$ and the mean velocity $u$. 
For certain interaction potentials only the divergence of the stress tensor appears (in weak form) in the resulting equations.

The aim of the present article
is to
show that the weak convergence of point measures and a uniform bound on their $(2+\epsilon)$-moments imply hydrodynamic equations at the limit of infinitely many interacting molecules and to provide examples satisfying the two assumptions. We also deal with two points that were not addressed in  \cite{D}: First, the existence of subsequences of point measures weakly convergent for all $t$ in some $[0,T]$. Second, the measurability in $t$ and the regularity of the macroscopic velocity $u$. We obtain  the macroscopic equations in section \ref{General}, first in Theorem \ref{MainMain}
using measures on $[0,T]\times\R^6$. Both issues of common in $t$ subsequences and the measurability of $u$ are then overcome. The Morrey assumptions used in \cite{D} satisfy the assumptions of Theorem \ref{MainMain}. Then in Theorem \ref{Main} we
consider $t$-families of measures on $\R^6$ that have common weakly convergent subsequences. As the assumptions of Theorem \ref{MainMain} are satisfied whenever the assumptions of Theorem \ref{Main} are, the measurability of $u$ is determined via comparison with Theorem \ref{MainMain}. 
We do not insist here on the form of the interaction part of the stress tensor -- for the examples that show later in Section \ref{PicardExample} this term does not appear at the limit.
Our main tools are the general Fubini theorem for families of measures, disintegration of measures, and convergence of measures. 

To substantiate the assumptions of the main theorems, we show that the conditions of Theorem \ref{Main}, and therefore of Theorem \ref{MainMain}, are satisfied at least when the solutions of the  classical equations for $N$ interacting molecules, rescaled as in \cite{M}, have accelerations uniformly bounded in $N$ on any finite time interval. This occupies the last subsection of  section \ref{General}. The main tool here is L\'evy continuity.
 
Finally, to provide explicit examples, we show in section \ref{PicardExample} that on any fixed time interval $[0,T]$, when the initial positions and velocities are bounded
and the rescaled interaction is weak enough compared to the initial conditions,  the accelerations are indeed uniformly bounded in $N$ on $[0,T]$. 
The choice of scale for each $N$ is such that distances between molecules can only increase in time. The main argument here is by continuity. (Under the same assumptions Picard iteration constructs the solutions on any finite time interval.)
These examples include, but are not restricted to, certain types of ``burst" configurations, i.e. examples where the velocities are any positive multiple of the positions at the initial moment in time. 

On the other hand, the choice of scale in section \ref{PicardExample} weakens the interaction and, for the examples in this section, the interaction term vanishes at the limit of infinitely many particles.
The stress tensor then consists only of velocity fluctuations. 
Therefore the hydrodynamic equations we obtain coincide with the macroscopic equations of Maxwell in \cite{Max}. In particular,  for the examples in this section we show that on any $[0,T]$ the momentum equation is precisely in the form established by Maxwell in \cite{Max} eq. (76), nowadays derived via the Boltzmann equation as in \cite{Gr}, eq. (2.46), for example\footnote{Recall that the validity of the Boltzman equation has been established only for a fraction of the collision time.}. As Maxwell argues, this form of the momentum equation can be approximated, up to certain order and for ``quiet'' flows, by the compressible Navier-Stokes equations.

\section{Hamiltonian equations} \label{Basics}

\subsection{Equations of motion}
We start with the motion of $N$ classical molecules, each of mass $m_N$, without external forces, and with pair interaction potential energy between a molecule at $x$ and a molecule at $y$ equal to $m_N^2  \Phi_N(|x - y|)$. Assuming always $\Phi_N$ of negative derivative for small distances, where molecules repulse, the force on a molecule at $x$ from a molecule at $y$ is $ -{m_N^2} \nabla_x \Phi_N(|x - y|)$, and the acceleration of the $i$-th molecule at time $t$, when its position is $x_i(t)$, satisfies
\begin{equation} \label{momentumeq}
\begin{split}
     &\ \ {m_N} u_i'(t)
     = 
     -m_N^2 \nabla_{x_i} 
     \sum_{\substack{j =1\\j\neq i}}^{N}
    \Phi_N\left(|x_i(t) - x_j(t)|\right).
\end{split}
\end{equation}
The total energy of the system consisting of these $N$ molecules when their positions and velocities are $x_i(t)$, $u_i(t)$ is 
\begin{equation}
\begin{split}
    E_N = \frac12
    {m_N}
     \sum_{i=1}^{N} |u_i(t)|^2
    +
    m_N^2
    \sum_{\substack{i,j =1\\i\neq j}}^{N}
    \Phi_N\left(|x_i(t) - x_j(t)|\right).
\end{split}
\end{equation}

\subsection{A length scale from the $N$-molecule system}
Recall that during a head-on collision between two molecules of mass $m_N$ (i.e. collision with impact parameter $0$) their minimum distance $\sigma_N$ satisfies $v_{\infty}^2 = {m_N}\Phi_N(\sigma_N)$, when the interaction potential is $m_N^2 \Phi_N(r)$ and for $v_{\infty}$ the molecules' relative speed at $t \to \pm \infty$, see \cite{LL}, \S 18.   
For $v_{\infty}$ independent of the $N$, this reads $m_N \Phi_N(\sigma_N) = \text{constant}$, and one way to accommodate this is to set
\begin{equation}  \label{MorreyScaling}
\begin{split}
     {m_N}\Phi_N(r) = \Phi\left( \frac{r}{\sigma_N}\right),
\end{split}
\end{equation}
cf. \cite{ABGS}.
This is what Morrey in \cite{M} adopts and this is what we shall also adopt again here.
Then \eqref{momentumeq} becomes
\begin{equation} \label{Hamp}
\begin{split}
    \frac{d x_i}{dt}
    &=
    u_i,
    \\ 
    \frac{du_i}{dt}
    &=
    -
    \sum_{\substack{j=1\\ j\neq i}}^{N}
    \frac{1}{\sigma_N}
    \Phi' 
    \left( 
    \frac{|x_i(t) - x_j(t)|}{\sigma_N}
    \right)
    \frac{x_i(t) - x_j(t)}
     {|x_i(t) - x_j(t)|},\quad
   1 \leq i \leq N.
\end{split}
\end{equation}

As energy is conserved, standard theory of ordinary differential equations, see for example \cite{CL}, p. 7, p. 47, and \cite{C}, p. 110, gives the following:
\begin{Thm}
Let ${\frak G}\subset \R^{3N}$ be
$   {\frak G}
    =
    \{
    (x_1,\ldots,x_N): x_i \in \R^3, i\neq j \Rightarrow x_i \neq x_j
    \} $.
Then if $\Phi$ has locally Lipshitz-continuous derivative on $(0,\infty)$ the  initial value problem for \eqref{Hamp} on ${\frak G}\times \R^{3N}$ has, for each $N$, unique solution on any time interval $[0,T]$.
\end{Thm}

\section{General results} \label{General}

\subsection{Standard measure theory}\
Weak convergence of measures is important in what follows: a sequence of measures $\mu_N$ converges weakly to a measure $\mu$, or $\mu_N \Rightarrow \mu$, if $  \displaystyle   \int f(x) \mu_N(dx) \to \int f(x) \mu(dx), 
     \ N \to \infty
$ for any $f$ bounded and continuous.   
The following is standard:
\begin{Lem}\label{pmom}
For any sequence of positive measures $\mu_N$ and $f$ measurable
\begin{equation}
\begin{split}
     \sup_N
     \int |f(x)|^{p + \varepsilon} \mu_N(dx)
     < 
     \infty
     \ \Rightarrow\ 
     \lim_{R\to \infty}
     \int\limits_{\{|f(x)|^p > R\}}
      |f(x)|^{p} \mu_N(dx) \to 0, \ \text{uniformly in $N$}.
\end{split}
\end{equation}
\end{Lem}
\begin{proof}
$\displaystyle
     R^{\varepsilon/p}
     \int\limits_{\{|f(x)|^p>R \}}
      |f(x)|^{p} \mu_N(dx)
      \leq
     \int\limits_{\{ |f(x)|^p>R \}}
      |f(x)|^{p+\varepsilon} \mu_N(dx)
      \leq
      \int
      |f(x)|^{p+\varepsilon} \mu_N(dx) . 
$
\end{proof}
Following \cite{AGS} Chapter 5, we call $f\geq 0$ uniformly integrable with respect to $\{\mu_N\}$ if
\begin{equation}
\begin{split}
     \lim_{R\to \infty}
     \int\limits_{\{ f >R \}}
      f(x) \mu_N(dx) \to 0, \ \text{uniformly in $N$}.
\end{split}
\end{equation}
Therefore Lemma \ref{pmom} provides a sufficient condition for $|f|^p$ to be uniformly integrable. 
The following shows as part of Lemma 5.1.7 in \cite{AGS}:
\begin{Lem} \label{continuousintegrand}
Let $\mu_N$ weakly converge to $\mu$ and $f$ continuous with $|f|$ uniformly integrable with respect to 
$\mu_N$. Then the $\mu_N$-integrals of $f$ converge (without passing to subsequence):
\begin{equation}
\begin{split}
     \int f(x) \mu_N(dx) \to \int f(x) \mu(dx), 
     \ N \to \infty. 
\end{split}
\end{equation}
\end{Lem}
We shall use repeatedly the following general Fubini and disintegration theorems. For proofs in a context relevant to this article see \cite{A}, \S 2.6 and \cite{AFP}, \S 2.5, respectively. All spaces in these theorems are $\R^k$ for some $k$, all $\sigma$-algebras are Borel, and a family of measures $\nu_{\lambda}$ on $X$ is Borel measurable if for any $B$ Borel set in $X$ the assignment $\lambda \mapsto \nu_{\lambda}(B)$ is measurable as a map from the $\lambda$'s to $\R$. To avoid completions of $\sigma$-algebras we always extend functions defined off a set of measure $0$ by setting them equal to $0$ on that set.

\begin{Thm}[General Fubini] \label{genFub}
Let $\mu(dx)$ be a probability measure on $(X, \mathcal{E})$ and
$\mu_{x}(dy)$ a measurable family of probability measures on $(Y, \mathcal{F})$.
Then there is unique probability measure $M$ on $(X\times Y, \mathcal{E} \times \mathcal{F})$ such that:
\begin{equation}
   \begin{split}
       M(B)
       =
       \int_X
       \left(
       \int_Y
       \chi_B(x,y)\mu_x(dy)
       \right)
       dx
   \end{split}
\end{equation}
and $M$ has the following properties:
\begin{enumerate} 
\item
for $f$ measurable and positive on $X\times Y$ the assignment
$\displaystyle     \int_Y f(x,y) \mu_x(dy)
$ defines an $x$-measurable function and 
$\displaystyle 
\int_{X\times Y} f(x,y) M(dx,dy) 
=
\int_X \int_Y f(x,y) \mu_x(dy) \mu(dx) \in [0,+\infty].
$ 
\item
for $f$ on $X\times Y$ such that $\displaystyle \int_{X\times Y} f(x,y) M(dx,dy)$ is finite then  $\displaystyle     \int_Y f(x,y) \mu_x(dy)$ exists for almost all $x$ and, once extended by $0$ to the remaining $x$'s, it
 defines an $x$-measurable function which satisfies 
$\displaystyle 
\int_{X\times Y} f(x,y) M(dx,dy) 
=
\int_X \int_Y f(x,y) \mu_x(dy) \mu(dx).
$ 
\end{enumerate}

\end{Thm} 
We shall use the notation
\begin{equation}
\begin{split}
      M(dx,dy)
     =
     \int 
     M_x(dy)\mu(dx)  
\end{split}
\end{equation} 
as a shorthand for the measure $M$ of Theorem \ref{genFub}.

\begin{Thm}[Disintegration] \label{dsntgrt}
For $M$ probability measure on $(X\times Y, \mathcal{E} \times \mathcal{F})$ 
and $\mu(dx) = (pr_1)_\#M (dx)$ on $(X,\mathcal E)$ there exists a $\mu$-almost-all uniquely determined Borel measurable family of probability measures $\{\mu_x(dy)\}$ on $(Y,\mathcal F)$ such that $\displaystyle M(dx,dy)=\int \mu_x(dy)\mu(dx)$.
%for each $x$ probability measure $\mu_x(dy)$ such that  the family $\mu_x(dy)$ is Borel measurable, $\mu$-almost everywhere uniquely defined, and such that items (1), (2), and (3) of the above Theorem hold.
\end{Thm} 

{\subsection{Molecule Measures}
For $x_i^{(N)}(t), u_i^{(N)}(t)$ solutions of the $N$-system \eqref{Hamp} on some fixed time interval $[0,T]$ with $T<\infty$,   
define for each $t$ 
the time dependent vector field
\begin{equation}
      u_N(t,x)=
      \left\{
     \begin{array}{ll}
    \dfrac{d x_i^{(N)}(t)}{dt}  & \text{if} \ x = x_i^{(N)}(t)\\
    \ \ \ \  0 & \text{o/w},
     \end{array}
     \right. 
\end{equation}
and 
the (molecule-velocity density) probability measures on $\R^6$ 
\begin{equation} \label{M_t^{N}}
\begin{split}
     M_t^{(N)}(dx,dv):=
     \frac1N
     \sum_{i=1}^{N} \delta_{\left(x_i^{(N)}(t), u_i^{(N)}(t)\right)}(d x,d v).
\end{split}
\end{equation}
The first marginal of $M_t^{(N)}$ is
\begin{equation}
\begin{split}
    \mu_t^{(N)}(d x)
    :=
    \frac{1}{N}
    \sum_{i =1}^{N} \delta_{ x_i^{(N)}(t)}(d x),
\end{split}
\end{equation}
or\footnote{For $h$ measurable  and $\nu$ measure, $h_\#\nu(B) = \nu(h^{-1}(B))$} 
    $M_t^{(N)}(dx, dv)
    =
    \left(Id \times u_N\right)_\#\mu_t^{(N)}(dx,dv)$.
Notice that for all $t$ the field $u_N(t,.)$ is defined for $\mu_t(dx)$-almost all $x$. 
The total mass being $1$ for all $N$, the factor $1/N$ is 
the mass of each molecule in the $N$-system.

The family $\left\{M_t^{(N)}(dx,dv):t\in[0,T]\right\}$ is Borel measurable: for any Borel $B \subset \mathbb R^6$, $\displaystyle t\rightarrow M_t^{(N)}(B)=\frac 1N\sum_{i=1}^{N}\chi_B\left(x_i^{(N)}(t),v_i^{(N)}(t)\right)$ is Borel since $\left(x_i^{(N)}(t),v_i^{(N)}(t)\right)$ is continuous in $t$ and $\chi_B$ is Borel. Then according to Theorem \ref{genFub} the measure 
\begin{equation} \label{definesM^{(N)}}
\begin{split}
      M^{(N)}(dt,dx,dv)
     =
     \int 
     M_t^{(N)}(dx,dv)dt  
\end{split}
\end{equation}
is well defined.
Conversely, given a probability measure $M(dt,dx,dv)$ and for $pr^{1,2}:(t,x,v) \mapsto (t,v)$, define
\begin{equation} \label{mu(dt,dx)}
\begin{split}
    \mu(dt,dx)
    :=
    (pr^{1,2}_{\#}M)(dt,dx,dv)
\end{split}
\end{equation} 
and disintegrate $M(dt,dx,dv)$ according to Theorem \ref{dsntgrt}  with respect to $\mu(dt,dx)$
\begin{equation}
    \begin{split}
        M(dt,dx,dv)
        =
        \int M_{t,x}(dv) \mu(dt,dx),
    \end{split}
\end{equation}
to get a $(t,x)$-Borel measurable family of measures $M_{t,x}(dv)$.
When $M(dt,dx,dv)$ has finite first moment the barycentric projection $u(t,x)$ given by 
\begin{equation} \label{defineu}
\begin{split}
     u(t,x) = \int v \,    M_{t,x}(dv)
\end{split}
\end{equation}
is, according to item (3) of Theorem \ref{genFub}, well-defined for $\mu$-almost all $(t,x)$ and, once extended by $0$ to the remaining $(t,x)$'s, it
 defines an $(t,x)$-measurable function which satisfies 
$\displaystyle 
\int_{[0,T]\times\R^6} v M(dt,dx,dv) 
=
\int_{[0,T]\times \R^3} u(t,x) \mu(dt,dx)$.

\begin{Lem}\label {limit measurable}
For each $N$ let $\left\{\displaystyle \nu_x^{(N)}(dy)\right\}$ be a  Borel measurable family of probability measures on $Y$ and let $\displaystyle \nu_x^{(N)}(dy)$ converge weakly to $\displaystyle \nu_x(dy)$ for all $x\in X$. Then for $\mu$ a probability measure on $X$ 
\begin{enumerate}
\item $\left\{\displaystyle \nu_x(dy)\right\}$ is a Borel measurable family, and
\item $\displaystyle \int \nu_x^{(N)}(dy)\mu(dx)\Rightarrow \int\nu_x(dy)\mu(dx)$.
\end{enumerate}
\end{Lem}
\begin{proof}
Let $\mathcal M=\left\{B\in \mathcal B(Y):x\to \nu_x(B) \text{ Borel measurable}\right\}$. For the first assertion, it is enough to show $ \mathcal M=\mathcal B(Y)$. By definition $\mathcal M\subset \mathcal B(Y)$, so it is enough to show  $\mathcal B(Y)\subset \mathcal M$.
Or, for  $\mathcal C=\left\{B \subset Y,\text{ closed}\right\}$ it is enough to show $\sigma(\mathcal C) \subset\mathcal M$.

First notice that $\mathcal M$ is closed under increasing limit and with respect to difference of sets: for  any increasing sequence in $\mathcal M$, $B_1\subset B_2\subset \ldots$ such that $B_n\to B$, $\displaystyle\nu_x(B_n)\to\nu_x(B)$. Therefore $\nu_x(B)$ is Borel, i.e. $B\in \mathcal M$. Also if both $A$ and $C$ are in $\mathcal M$ and $A\subset C$, we have $\nu_x(C\backslash A)=\nu_x(C)-\nu_x(A)$, therefore $\nu_x(C\backslash A)$ is Borel and $C\backslash A\in \mathcal M$. 

Next we show $\mathcal C\subset \mathcal M$. For any $B \in \mathcal C$, approximate  $f_n(y)\rightarrow \chi_B(y)$, $n  \to \infty$  by $f_n$ positive, continuous and bounded (e.g. $f_n(y) =(1+nd\left(y,B\right))^{-1}$).
Then  
     \begin{equation}
         \begin{split}
             \nu_x(B)
             =
             \int \chi_B \nu_x(dy)
             =
             \lim_{n\rightarrow \infty} 
             \int f_n \nu_x(dy)
             =
             \lim_{n\rightarrow \infty}
             \lim _{N\rightarrow \infty}
             \int
              f_n \nu_x^{(N)}(dy).
         \end{split}
     \end{equation}
By assertion (2) in Theorem \ref{genFub} (General Fubini), $\displaystyle  \int f_n \nu_x^{(N)}(dy)$ is Borel, therefore so is $\nu_x(B)$.

It is clear that $\mathcal C$ is closed under finite intersections and $Y\in \mathcal C$. Then by the Monotone Class Theorem, \cite{JP}, p.36, $\sigma(\mathcal C) \subset  \mathcal M $.

For the second assertion note that for any $f(x,y)$ bounded continuous,
\begin{equation}
   \begin{split}
       \lim_{N\to \infty}
       \iint f(x,y)\nu_x^{(N)}(dy)\mu(dx)
              &=
       \int \lim_{N\to \infty} 
    \left( \int f(x,y)\nu_x^{(N)}(dy)\right)\mu(dx)\\
    &=
    \int
    \left( \int f(x,y)\nu_x(dy)\right)\mu(dx)\\
    &=
    \iint f(x,y)\nu_x(dy)\mu(dx).
   \end{split}
\end{equation}
\end{proof}
\begin{Rmk}
Suppose $\displaystyle \nu_x^{(N)}(dy)$ converges weakly to $\displaystyle \nu_x(dy)$ for $\mu$-a.e $x\in X$. Then Lemma \ref{limit measurable} holds for $\nu_x(dy)$ extending trivially.
\end{Rmk}

\subsection{Interaction Terms} Define now for $\Phi$ from \eqref{Hamp} and for any $\varphi \in C_0^{\infty}((0,T)\times\R^3)$
\begin{equation}
\begin{split}
    I_{\Phi}^{(N)} (t,\varphi) :=
    -\frac1{N\sigma_N}
    \sum_{i=1}^{N} 
    \varphi(t, x_i(t))
    \sum_{\substack{j=1\\ j \neq i}}^{N} 
    \Phi'\left(\frac{| x_i (t)-  x_j(t)|}{\sigma_N}\right)
    \frac{ x_i(t) -  x_j(t)}{| x_i(t) -  x_j(t)|},
\end{split}
\end{equation}

\begin{equation}
\begin{split}
     {\frak I}_{\Phi}^{(N)} (\varphi)
%     &:=
%     -\frac1{N\sigma_N}
%    \int_0^T
%    \sum_{i=1}^{(N)} 
%    \varphi(t, x_i(t))
%    \sum_{\substack{j=1\\ j \neq i}}^{(N)}
%    \Phi'\left(\frac{| x_i (t)-  x_j(t)|}{\sigma_N}\right)
%    \frac{ x_i(t) -  x_j(t)}{| x_i(t) -  x_j(t)|}
%    \,dt
%%    \\
%%    &
%    \left(
     =
    \int_0^T I_{\Phi}^{(N)}(t,\phi)\,dt
    %\right)
    ,
\end{split}
\end{equation}
and 
\begin{equation}
\begin{split}
     {\frak I}_{\Phi}(\varphi): = 
     \lim_{N\to \infty}
     {\frak I}_{\Phi}^{(N)} (\varphi), 
\end{split}
\end{equation}
when the limit exists. We will not insist on the form of the $\frak I_{\Phi}$ term here. \cite{D} shows how $\frak I_{\Phi}$ can be weakly of the form $\div S$ for interaction potentials without forces close to the center of the interaction, cf. \cite{G}, p. 110. $\frak I_{\Phi}$  vanishes for the examples that follow here in section \ref{PicardExample}. Other forms of $\frak I$ and its role will appear elsewhere.  

\subsection{The Main Theorems} We are now ready to prove the main theorems on the hydrodynamic equations at the limit of infinitely many molecules.

\begin{Thm} \label{MainMain}
Assume that
\begin{enumerate}
\item  
        $M^{(N)}(dt,dx,dv)$ as in \eqref{definesM^{(N)}} converge weakly to $M(dt,dx,dv)$
and
\item
for some $\varepsilon>0$
\begin{equation}\label{Ass2new}
   \begin{split}
        \sup_N 
        \int \left| v\right|^{2+\varepsilon}
         M^{(N)}(dt,dx,dv)
         <
         \infty.
   \end{split}
\end{equation}

\end{enumerate}
Then the barycentric projection $u$ of $M(dt,dx,dv)$  
as in \eqref{defineu} satisfies $u \in  L^{2+\varepsilon}([0,T]\times \R^3, \mu)$  for $\mu(dt,dx)$ as in \eqref{mu(dt,dx)}, and for $\mu_t$ such that $\displaystyle \mu = \int \mu_t(dx) dt$ and any $\varphi \in C_0^{\infty}((0,T)\times \R^3)$  the following continuity  and momentum equations hold:  
\begin{equation}
\begin{split}\label{ContinuityEqn}
     &\int_0^T \int_{\R^3}
     \left\{
     \frac{\partial \varphi}{\partial t} (t,x)
      +
      \nabla_x \varphi(t,x) \cdot  u(t,x)
      \right\}
      \mu_t(dx)\ dt
     =
     0,
\end{split}
\end{equation}
\begin{equation}\label{MomentumEqn}
\begin{split}
    & \int_0^T\int_{\R^3}
   \frac{\partial \varphi}{\partial t}
    (t, x)  u (t, x)  \mu_t(d x) dt
   +
   \int_0^T\int_{\R^3}
   \nabla \varphi(t, x)\cdot  u(t, x)   u(t, x) \mu_t (d x)  dt 
    \\
    &
    =
    -
    \int_0^T\int_{\R^3}
    \nabla \varphi(t,x) 
    \cdot\int ( v - u(t, x) )
    ( v -  u(t, x))\, M_{t, x}(dv)\ \mu_t (d x) dt
    \ +\
    {\frak I}_{\Phi}(\varphi).
\end{split}
\end{equation}
\end{Thm}

\begin{proof}
For any $\varphi(t,x) \in C_0^{\infty}((0,T)\times \R^3)$, by the first of \eqref{Hamp}
\begin{equation}
    \begin{split}
         \dfrac{d}{dt}\sum_i \varphi\left(t,x_i^{(N)}\right)
         &=
         \sum_i \partial_t\varphi\left(t,x_i^{(N)}\right)+\nabla_x\varphi\left(t,x_i^{(N)}\right)u_i^{(N)}
         =
         \int_{\R^6}
          \left(
          \partial_t
          \varphi(t,x)+\nabla_x\varphi(t,x)v
           \right)
         M_t^{(N)}(dx,dv)
        \end{split}
\end{equation}
and by the second of \eqref{Hamp}
\begin{equation}
    \begin{split}
         \dfrac{d}{dt}\sum_i \varphi\left(t,x_i^{(N)}\right)u_i^{(N)}
         &=
         \sum_i \partial_t\varphi\left(t,x_i^{(N)}\right)u_i^{(N)}+\nabla_x\varphi\left(t,x_i^{(N)}\right)\cdot u_i^{(N)}u_i^{(N)}
         +
         I_{\Phi}^{(N)}(\varphi)\\
         &=
         \int_{\R^6}
         \left(
          \partial_t\varphi(t,x)v
          +
          \nabla_x\varphi(t,x)\cdot v \,v
          \right) M_t^{(N)}(dx,dv)
         +
         I_{\Phi}^{(N)}(\varphi).
        \end{split}
\end{equation}
Integrating over $[0,T]$, and since $\varphi(t,x)$ has compact support, the left hand sides of these equations integrate to $0$. This gives the continuity and momentum equations for each  $N$-system:
\begin{equation} \label{Ncontinuity}
   \begin{split}
       \int\limits_{[0,T]\times\R^6}
        \left(
         \partial_t\varphi(t,x)+\nabla_x\varphi(t,x)v
         \right)
         M^{(N)}(dt,dx,dv)
        =
        0,
   \end{split}
\end{equation}
\begin{equation}  \label{Nmomentum}
    \begin{split}
          \int\limits_{[0,T]\times\R^6}
         \left(
         \partial_t\varphi(t,x)v
         +
         \nabla_x\varphi(t,x)\cdot vv
         \right)
         M^{(N)}(dt,dx,dv)
\         +
         {\frak I}_{\Phi}^{(N)}(\varphi)
        =
        0.
    \end{split}
\end{equation}
The first marginal $M^{(N)}(dt,dx,dv)$ is $dt$ for all $N$. Therefore the first marginal of $M(dt,dx,dv)$, and hence of $\mu(dt,dx)$,  is also $dt$. 
 By the weak convergence of $M^{(N)}(dt,dx,dv)$  and the definition of $\mu_t$
\begin{equation}
\begin{split}
        \int\limits_{[0,T]\times\R^6}
         \partial_t\varphi(t,x)
         M^{(N)}(dt,dx,dv)
        & \rightarrow
        \int\limits_{[0,T]\times\mathbb R^6}
        \partial_t\varphi(t,x)
       M(dt,dx,dv)\\
       &=
       \int\limits_{[0,T]\times\mathbb R^3}
        \partial_t\varphi(t,x)
        \int_{\mathbb R^3}
       M_{t,x}(dv)\mu(dt,dx)\\
       &=
       \int\limits_{[0,T]\times\mathbb R^3}
        \partial_t\varphi(t,x)
        \mu(dt,dx)\\
        &=
        \int_0^T
        \int_{\mathbb R^3}
        \partial_t\varphi(t,x)
        \mu_t(dx)dt.
\end{split}
\end{equation}
In addition, from \eqref{Ass2new} 
\begin{equation}
\begin{split}
        \sup_N
        \int\limits_{[0,T]\times\R^6} 
        |\nabla_x\varphi(t,x)\, v|^{2+ \varepsilon}
         M^{(N)}(dt,dx,dv)
       &\leq
       C_{\varphi}
       \sup_N
        \int\limits_{[0,T]\times\R^6} 
        |v|^{2+ \varepsilon}
         M^{(N)}(dt,dx,dv)
        < \infty, 
\end{split}
\end{equation}
and similarly for integrals involving $\partial_t \varphi$ instead of $\nabla_x \varphi$. Then, using Lemma
 \ref{continuousintegrand},
definition \eqref{defineu}, and the definition of $\mu_t$,
\begin{equation}
   \begin{split}
        \int\limits_{[0,T]\times\R^6}
        \nabla_x\varphi(t,x)v \,
         M^{(N)}(dt,dx,dv)
        &\rightarrow
        \int\limits_{[0,T]\times\mathbb R^6}
        \nabla_x\varphi(t,x)v M(dt,dx,dv)\\
        &=
         \int_0^T
        \int_{\mathbb R^3}
        \nabla_x\varphi(t,x)
        u(t,x)
        \mu_t(dx)dt,
        \end{split}
\end{equation}
\begin{equation} \label{firsttermconverge}
   \begin{split}
      \int\limits_{[0,T]\times\R^6}
       \partial_t\varphi(t,x)
        v 
        M^{(N)}(dt,dx,dv)
       & \rightarrow
        \int\limits_{[0,T]\times\mathbb R^6}
        \partial_t\varphi(t,x)v M(dt,dx,dv)\\
        &=
         \int_0^T
        \int_{\mathbb R^3}
        \partial_t\varphi(t,x)u(t,x)
        \mu_t(dx)dt,
   \end{split}
\end{equation}
\begin{equation} \label{secondtermconverge}
\begin{split}
      \int\limits_{[0,T]\times\R^6}
         \nabla_x\varphi(t,x)\cdot v v 
         M^{(N)}(dt,dx,dv)
        \rightarrow&
         \int\limits_{[0,T]\times\mathbb R^6}
        \nabla_x\varphi(t,x)\cdot v v M(dt,dx,dv).
\end{split}
\end{equation}
Adding and subtracting $u(t,x)$, the last limit can be rewritten as
\begin{equation}
    \begin{split}
         &\int\limits_{[0,T]\times\mathbb R^6}
        \nabla_x\varphi(t,x)\cdot v v M(dt,dx,dv)\\
        =&
        \int\limits_{[0,T]\times\mathbb R^3}
        \nabla_x\varphi(t,x)\cdot 
        \int_{\mathbb R^3}
        \{(v-u(t,x))(v-u(t,x)) + u(t,x)u(t,x)
        \\
        & \ \ \ \ \ \ \ \ \ \ \ \ \ \ \ \ \ \ \ \ \ \ 
        \ \ \ \ \ \ \ \ \ \ \ +
        (v-u(t,x)) u(t,x) + u(t,x)(v-u(t,x))\}
        M_{t,x}(dv)\mu(dt,dx),
\end{split}
\end{equation}
where the last two  terms in the $M_{t,x}$ integrand integrate to zero.
For the remaining terms, notice that $M(dt,dx,dv)$ has finite $v$-moment by \eqref{Ass2new}, and therefore, according to the remarks following \eqref{defineu}, $u(t,x)$ is measurable. In addition, 
\begin{equation} \label{2+}
\begin{split}
      \int\limits_{[0,T]\times\mathbb R^3}
        \left| 
        u(t,x)\right|^{2+\varepsilon}
        \mu(dt,dx)  
        &=
      \int\limits_{[0,T]\times\mathbb R^3}
        \left| 
        \int vM_{t,x}(dv)\right|^{2+\varepsilon}
        \mu(dt,dx)\\
        &\leq
      \int\limits_{[0,T]\times\mathbb R^3}
        \int |v|^{2+ \varepsilon}M_{t,x}(dv)
        \mu(dt,dx) \\
        &=
        \int\limits_{[0,T]\times\mathbb R^6}
         |v|^{2+\varepsilon}M(dt,dx,dv)\\
        &\leq
        \liminf_N \int\limits_{[0,T]\times\mathbb R^6}
         |v|^{2+\varepsilon}M^{(N)}(dt,dx,dv) < \infty, 
     \end{split}
\end{equation}
where in the last step we used the lower semicontinuity of weak convergence (valid for all lower semicontinuous and bounded below functions). Therefore, $u \in  L^{2+\varepsilon}([0,T]\times \R^3, \mu)$ and 
$\nabla_x\varphi(t,x)\cdot (v-u(t,x))(v-u(t,x))$ and 
$\nabla_x\varphi(t,x)\cdot u(t,x)u(t,x)$ are separately integrable
with respect to $M$. We can then write
\begin{equation}\label{vv break}
\begin{split}
      \int\limits_{[0,T]\times\mathbb R^6}
        \nabla_x\varphi(t,x)\cdot v v
         &M(dt,dx,dv)
        \\
        =
        &
        \int_0^T\int_{\R^3}
        \nabla \varphi(t, x)\cdot  u(t, x)   u(t, x) \mu_t (d x)  dt  \\ 
    &\ \ \ \ \ \  \ 
    +\int_0^T\int_{\R^3}
    \nabla \varphi(t,x) 
    \cdot\int ( v - u(t, x) )
    ( v -  u(t, x))\, M_{t, x}(dv)\ \mu_t (d x) dt.
    \end{split}
\end{equation}
The existence of the two limits \eqref{firsttermconverge}, \eqref{secondtermconverge} and equation \eqref{Nmomentum} imply that ${\frak I}_{\Phi}^{(N)}(\varphi)$ also converges.
\end{proof}

\begin{Rmk}
It is standard that for $\displaystyle u(t,x)\in L^1\left([0,T]\times \mathbb R^3, \mu \right)$ satisfying the continuity equation \eqref{ContinuityEqn} there is weakly continuous in $t$ Borel family $\{\widetilde{\mu}_t(dx)\}$ such that $\widetilde{\mu}_t(dx)=\mu_t(dx)$ for almost all $t$, see \cite{AGS}, Lemma 8.1.4.  
\end{Rmk}

%%%%%%%%%%%%%Morrey Assumptions%%%%%%%%%%%%%%%%%%%%%%%%%%%%%%%%
\begin{Exam}\label{tightness of Mdt}
The sequence of measures
     $M^{(N)}(dt,dx,dv)$
on $[0,T]\times \R^6$ has weakly convergent subsequence whenever there is finite constant $B$ such that for almost all $t$
\begin{equation} \label{MorreyAss}
\begin{split}
    \frac1N\sum_{i=1}^{N} \left|x_i^{(N)}(t)\right|^2 < B, 
    \quad
    \frac1N\sum_{i=1}^{N} \left|u_i^{(N)}(t)\right|^2 < B:
\end{split}
\end{equation}
 by Chebyshev's inequality, 
\begin{equation} 
    \begin{split}
        \iint \limits_{|(t,x,v)|>R}
        M_t^{(N)}(dx,dv)dt
        &\leq
        \int_0^T\int_{\mathbb R^6} 
        \dfrac {t^2+|x|^2+|v|^2}{R^2}M_t^{(N)}(dx,dv)dt\\
        &\leq
        \dfrac {1}{R^2}\left(2B T+\dfrac {T^3}{3}\right).
    \end{split}
\end{equation}
The claim now follows from Prohorov's criterion, see \cite{GS}, p. 362.

Conditions \eqref{MorreyAss} are satisfied whenever the same inequalities hold at $t = 0$ and the energies are uniformly bounded in  $N$, see \cite{M}, Theorem 5.2. See also \cite{D}, Proposition 3.1.  
\end{Exam}
%%%%%%%%%%%%%%%%%%%%%%%%%%%%%%%%%%%%%%%%%%%%%%%%%%%%%%%%%%%%%%%%%%%%%%

The following shows that pointwise convergence and bounds
give the same results as in the previous theorem. This version is closer to the main result in \cite{D}.

Whereas Theorem \ref{MainMain} is useful for describing
the measurability of the $u$ and its assumptions are weaker
than those of Theorem \ref{Main} that follows, the pointwise convergence of measures  $M_t^{(N)}(dx, dv)
     \Rightarrow
      \widehat{M}_t(dx, dv)$  is perhaps of more interest for applications.
%In Theorem \ref{Main}, $u$ and $\mu_t$ are defined by Theorem \ref{MainMain} for the measure $M = \displaystyle \int \widehat{M}_t\, dt$.
\begin{Thm} \label{Main}
Let $T>0$ be any fixed finite time.
Assume that there is subsequence of positive integers, and some $\epsilon >0$ such that for all  $N$'s of that subsequence and for almost all $t$ in $[0,T]$
the following hold: 
\begin{enumerate}
\item
$\displaystyle
     M_t^{(N)}(dx, dv)
     \Rightarrow
      \widehat{M}_t(dx, dv), 
     \ N \to \infty$,
\item
$\displaystyle\sup_N \int |u_N|^{2 + \epsilon}(t,x) \mu_t^{(N)}(dx) 
     <
    \infty. $
\end{enumerate}
Let $\widehat{\mu}_t$ be the first marginal of $\widehat{M}_t$ and $\widehat{u}(t,x)$ the corresponding barycentric projection. Then $\displaystyle \int \int |\widehat{u}|^{2 + \epsilon} (t,x) \widehat{\mu}_t(dx) dt < \infty$ and the continuity equation \eqref{ContinuityEqn} and the momentum equation \eqref{MomentumEqn} hold for $\widehat{\mu}_t$ and $\widehat{u}(t,x)$.
%Then the barycentric projection $\widehat{u}(t,.)$ of $\widehat{M}_t$ with respect to its first marginal $\widehat{\mu}_t$ satisfies $\widehat{u}(t,.)\in L^{2+\epsilon}\left({\mu}_t\right)$ for almost all $t$, $\displaystyle \int \int |\widehat{u}|^{2 + \epsilon} (t,x) \mu_t(dx) dt < \infty$, $u = \widehat u$ for $\mu_t$-almost all $x$ for almost all $t$, and  the continuity equation   and the momentum equation hold.
\end{Thm}
\begin{proof}

Lemma \ref{limit measurable}  and the first assumption here imply $\displaystyle \int M_t^{(N)}(dx,dv)dt
\Rightarrow  \int \widehat{M}_t(dx,dv)dt =: M(dt,dx,dv)$. This is  the first assumption of Theorem \ref{MainMain}. Also the second assumption here implies the second assumption of Theorem \ref{MainMain}. Therefore the continuity and momentum equations are satisfied for $u$ and $\mu_t$ as defined in Theorem \ref{MainMain} for $M$. 

By definition of $M_{t,x}(dv)$ in  Theorem \ref{MainMain}, $\left\{M_{t,x}(dv)\right\}$ is Borel measurable in $(t,x)$, therefore for fixed $t$ Borel measurable in $x$. Then $M_t(dx,dv):=\displaystyle \int  {M}_{t,x}(dv) \mu_t(dx) $ is well defined. Next we show $\widehat{M}_t=M_t$ for all most $t$. 
Notice that for any $B$ Borel in $[0,T]\times \mathbb R^6$, by the definitions of $M_{t,x}(dv)$ and $\mu_t(dx)$
\begin{equation}
\begin{split}
     \int \chi_B M(dt,dx,dv)
     &=
     \int 
     \left(
     \int  \chi_B {M}_{t,x}(dv)
     \right)
     \mu(dt,dx)\\
     &=
     \int
     \left(
     \int 
     \left(
     \int \chi_B {M}_{t,x}(dv)
     \right)
     \mu_t(dx)
     \right)
     dt\\
     &=
     \int
     \left(
     \int
      \chi_B M_t(dx,dv)
     \right)
      dt.
\end{split}
\end{equation}
%Notice that for any $f(t,x,v)$ positive and measurable in $(t,x,v)$
%\begin{equation}
%\begin{split}
%     \int f(t,x,v) M(dt,dx,dv)
%     &=
%     \int 
%     \left(
%     \int f(t,x,v) {M}_{t,x}(dv)
%     \right)
%     \mu(dt,dx)\\
%     &=
%     \int
%     \left(
%     \int 
%     \left(
%     \int f(t,x,v) {M}_{t,x}(dv)
%     \right)
%     \mu_t(dx)
%     \right)
%     dt,
%\end{split}
%\end{equation}
%where $\displaystyle \int f(t,x,v) {M}_{t,x}(dv)$ is $(t,x)$-measurable and therefore $\displaystyle \int  \left(\int f(t,x,v) {M}_{t,x}(dv) \right)\mu_t(dx)$ is $t$-measurable. Then for any $B$ Borel in $\mathbb R^6$,
%\begin{equation}
%    \begin{split}
%       M_t(B)
%       =
%       \int_{\mathbb R^3} \int_{\mathbb R^3} \chi_B M_{t,x}(dv)\mu_t(dx)
%    \end{split}
%\end{equation}
%is Borel in $t$. 
%Therefore $\left\{ M_t(dx,dv) \right\}$ is Borel in $t$ and 
Therefore $M(dt,dx,dv)=\displaystyle \int M_t(dx,dv)dt$. The uniqueness of the disintegration of $M$ with respect to $dt$ gives $\widehat{M}_t(dx,dv)=M_t(dx,dv)$ for almost all $t$, therefore $\widehat{\mu}_t(dx)=\mu_t(dx)$ for almost all $t$. Now the uniqueness of the disintegration of $\widehat{M}_t(dx,dv)$ with respect to $\widehat{\mu}_t(dx)$ gives $M_{t,x}(dv)=\widehat{M}_{t,x}(dv)$ for $dt$-almost all $t$ and $\widehat{\mu}_t$-almost all $x$. This implies that $\widehat{u}(t,x)=u(t,x)$ for $dt$-almost all $t$ and $\widehat{\mu}_t$-almost all $x$. Then the assertions of the Theorem follow from their counterparts in Theorem \ref{MainMain} for $u(t,x)$ and $\mu_{t}(dx)$. 
\end{proof}

\subsection{Subsequences for bounded velocities and bounded accelerations} \label{bounded=>assumptionssatisfied}

The second assumption of Theorem \ref{Main} clearly holds whenever velocities are uniformly bounded in $N$ for all $t$ in $[0,T]$ since $\mu_t^{(N)}$ are probability measures for all $N$ and $t$.
We now show that if the solutions of the $N$-Hamiltonian systems have velocities and accelerations uniformly bounded in $N$ for all $t \in [0,T]$ then the first assumption of Theorem \ref{Main} also holds. The main tool for this is the use of characteristic functions rather than the Prohorov's criterion, cf. \cite{M}, p. 291. We then use L\'evy continuity, which holds without passing to subsequences, for the convergence of measures.

\begin{Prop}
Assume uniformly bounded in $N$ accelerations (and therefore velocities and positions): 
\begin{equation}
\begin{split}
    \left|
    x_i^{(N)}(t)
    \right|,
    \left|
    u_i^{(N)}(t)
    \right|, 
    \left|
    \frac{d}{dt}u_i^{(N)}(t)
    \right|
    \leq B_T, 
    \quad t \in [0,T], 
    \quad i = 1, \ldots, N, 
    \quad N\in \N.
\end{split}
\end{equation}
Then there 
is a subsequence $N_i$ of $N$'s, independent of $t$, such that for all $t$ in $[0, T]$, 
\begin{equation}
\begin{split}
     M_t^{(N_i)}(dx,dv) 
     \Rightarrow 
     M_t(dx,dv).
\end{split}
\end{equation}
\end{Prop}
\begin{proof}
The characteristic of $M_t^{(N)}$ is \begin{equation}
\begin{split}
     \psi_N(t, y, w)
     &=
     \int e^{\displaystyle i(y\cdot x + w\cdot v)}
     \, M_t^{(N)}(dx, dv)\\
     &=
     \frac1N\sum_{j =1}^{N}
     e^{\displaystyle i\,(\, y\cdot x_j^{(N)}(t) + w\cdot u^{(N)}_j(t)\,)}.
\end{split}
\end{equation}
Therefore
\begin{equation}
\begin{split}
   \partial_t \psi_N(t, y, w)
   &=
   i \frac1N
   \sum_{j = 1}^{N}
     e^{\displaystyle i\,\left(\, y\cdot  x_j^{(N)}(t) + w\cdot  u^{(N)}_j(t)\,\right)}
     \left(
     \, y\cdot u_j^{(N)}(t) + w\cdot \left(u^{(N)}_j\right)'(t)
     \,\right),
     \\
    \nabla_y \psi_N(t, y, w) 
    &=
    i\frac1N\sum_{j = 1}^{N}
     e^{\displaystyle i\,\left(\, y\cdot x_j^{(N)}(t) + w\cdot u^{(N)}_j(t)\,\right)}
     x_j^{(N)}(t),
     \\
     \nabla_w \psi_N(t, y, w) 
     &=
     i\frac1N\sum_{j=1}^{N}
     e^{\displaystyle i\,\left(\, y\cdot  x_j^{(N)}(t) + w\cdot  u^{(N)}_j(t)\,\right)}
     u^{(N)}_j(t) ,         
\end{split}
\end{equation}
are all bounded.
In particular, for each fixed $T$ and $k \in \N$, there is  
uniformly convergent subsequence of $\psi_N(t,y)$ on $[0,T] \times B_k(0)$ by Arzela-Ascoli. Therefore, by taking $k \to \infty$  and diagonalizing, there is subsequence $\psi_{N_i}$ which converges for all $t$ and $y$ in $[0,T] \times \R^3$ (and which still converges uniformly on any $[0,T] \ \times$ compact). The limit is, of course, continuous in $t$, $y$, and $w$ as the uniform limit of continuous functions.
Apply now the L\'evy continuity theorem \cite{JP}, p. 167, for any fixed $t$ on this subsequence to find that, without resorting to any further subsequence, $M_t^{N_i}(dx,dv)$ converges weakly for all $t$. 
\end{proof}

{
\section{Solutions with bounds uniform in $N$.} \label{PicardExample}

For each {fixed $N$} the solutions of the $N$-system stay of course bounded on finite time intervals, by continuity. Given uniform bounds on the energy and the initial conditions,  the averages 
$\displaystyle \frac1N \sum \left |x_i^{(N)}(t)\right |^2$, 
$\displaystyle \frac1N \sum \left |u_i^{(N)}(t)\right |^2$ 
 stay uniformly bounded in $N$ on finite time intervals, cf. \cite{M}, Theorem 5.2. 
Here we show that there exists a class of examples where the {solutions themselves stay bounded uniformly in $N$ on any finite interval. This class then satisfies all assumptions of Theorem \ref{Main}. 

We are interested then in uniform in $N$ estimates for the
system of ODEs (where the solutions now show their dependence on $N$):
\begin{equation} \label{Ham}
\begin{split}
    \frac{d x_i^{(N)}}{dt}
    &=
    u_i^{(N)},
    \\ 
    \frac{du_i^{(N)}}{dt}
    &=
    -
    \sum_{\substack{j=1\\ j\neq i}}^{N}
    \frac{1}{\sigma_N}
    \Phi' 
    \left( 
    \frac{\left|x_i^{(N)}(t) - x_j^{(N)}(t)\right|}{\sigma_N}
    \right)
    \frac{x_i^{(N)}(t) - x_j^{(N)}(t)}
     {\left|x_i^{(N)}(t) - x_j^{(N)}(t)\right|},\quad
   1 \leq i \leq N,
\end{split}
\end{equation}
with $\sigma_N \to 0$. 

\begin{Thm} \label{PicardThm}
Assume $\Phi'$ decreasing on $(0, \infty)$,
$\left\{x_i(0), u_i(0)\right\}_{i \in \N}$ satisfying
\begin{equation}
\begin{split}
     \left|x_i(0)\right| \leq X, \quad
     \left | u_i(0) \right | \leq U, \quad \text{for \ all} \quad i\in \mathbb N,
\end{split}
\end{equation}
$B_N$ such that for $X_{ij} := x_i(0) - x_j(0)$ and $U_{ij} := u_i(0) - u_j(0)$
\begin{equation} \label{B_NCondition}
\begin{split}
     X_{ij}\cdot U_{ij}
     -
     2T\left|X_{ij}\right|B_N
     -
     3T^2\left |U_{ij}\right |B_N-
     2T^3B_N^2
     \geq
     0,
\end{split}
\end{equation}
and $\sigma_N$ such that
\begin{equation} \label{sigma_N}
\begin{split}
    -
    \frac{1}{\sigma_N}
    \sum_{\substack{j = 1\\ j \neq i}}^{N}
    \Phi' 
    \left( 
    \frac{\left|X_{ij}\right|}{\sigma_N}
    \right) 
    < B_N,\ \text{for all} \ i = 1, \ldots, N.
\end{split}
\end{equation}
Then
the solutions of \eqref{Ham} with initial conditions $\{x_i(0),u_i(0)\}_{i=1, \ldots,N}$  satisfy
\begin{equation}\label{Bounds}
\begin{split}
    \left|x_i^{(N)}( t) - x_j^{(N)}(t)\right|\ \big \uparrow\ \text{in} \ t \in [0,T] , \\
    \left| \frac{d}{dt}u_i^{(N)}(t)\right|
    \leq
    B_N, \quad t \in [0,T], 
\end{split}
\end{equation} 
(and, therefore
$ \left|x_i^{(N)}(t)\right|
   \leq
    X
     +
    UT+B_NT^2$,
    $\left|u_i^{(N)}(t)\right|
    \leq
    U + B_N T $, 
    $t \in [0,T])$.
\end{Thm}

\begin{Rmk} It is necessary, from \eqref{B_NCondition} that 
the increments of positions and velocities are ``alligned" in the sense that
\begin{equation} \label{align}
\begin{split}
    X_{ij} \cdot U_{ij}
    =
     \left(x_i(0)
    -
    x_j(0)\right)
    \cdot
    \left(
      u_i(0) - u_j(0)
      \right)> 0. 
\end{split}
\end{equation}
\end{Rmk}

\begin{proof}[Proof of Theorem \ref{PicardThm}]
Let 
\begin{equation}
   \begin{split}
       A_i^{(N)}(t)
       =
       -
    \sum_{\substack{j=1\\ j\neq i}}^{N}
    \frac{1}{\sigma_N}
    \Phi' 
    \left( 
    \frac{\left|x_i^{(N)}(t) - x_j^{(N)}(t)\right|}{\sigma_N}
    \right)
    \frac{x_i^{(N)}(t) - x_j^{(N)}(t)}
     {\left|x_i^{(N)}(t) - x_j^{(N)}(t)\right|},
   \end{split}
\end{equation}
and
\begin{equation}
   \begin{split}
       F_i^{(N)}(t)
       =
       -
    \sum_{\substack{j=1\\ j\neq i}}^{N}
    \frac{1}{\sigma_N}
    \Phi' 
    \left( 
    \frac{\left|x_i^{(N)}(t) - x_j^{(N)}(t)\right|}{\sigma_N}
    \right).
   \end{split}
\end{equation}
Then 
\begin{equation}
   \begin{split}
       \left|A_i^{(N)}(t)\right|
       \leq
       F_i^{(N)}(t).
   \end{split}
\end{equation}
Suppose $t=t_N$ is the first time such that $F_i^{(N)}(t_N)=B_N$. By continuity, $t_N>0$. Then for $0\leq t\leq t_N$,
\begin{equation} \label{A<B}
    \begin{split}
         \left|A_i^{(N)}(t)\right|
        \leq
        B_N,
    \end{split}
\end{equation}
and
\begin{equation}
\begin{split}
     &\dfrac12 \dfrac{d}{dt}
     \left|x_i^{(N)}(t)-x_j^{(N)}(t)\right|^2
     =
     \left(x_i^{(N)}(t)-x_j^{(N)}(t)\right)
     \cdot
     \left(u_i^{(N)}(t)-u_j^{(N)}(t)\right)\\
     =&
     \left(
     X_{ij}
     +
     \int_0^t 
     \left(\int_0^s \left(A_i^{(N)}(q)-A_j^{(N)}(q)\right) dq+U_{ij}\right)
     ds
     \right)
     \cdot
     \left(
     U_{ij}+\int_0^t\left(A_i^{(N)}(s)-A_j^{(N)}(s)\right)ds
     \right),
\end{split}
\end{equation}
after using the equations of motion. Expanding and using \eqref{A<B}, this is estimated from below by
\begin{equation}
\begin{split}
    &X_{ij}\cdot U_{ij}
     -
     \left|U_{ij}\right|
     \int_0^t \int_0^s 2B_N dqds
     -
     \left|X_{ij}\right|
     \int_0^t2B_N ds
     -
     t
     \left | U_{ij}\right |
     \int_0^t2B_N ds
     -
     \int_0^t2B_N ds
     \int_0^t\int_0^s2B_Ndqds\\
     &\geq
     X_{ij}\cdot U_{ij}
     -
     2T\left|X_{ij}\right|B_N
     -
     3T^2\left |U_{ij}\right | B_N-
     2T^3B_N^2
     ,
\end{split}
\end{equation}
which, by \eqref{B_NCondition}, is positive.
Then $|x_i(t)-x_j(t)|$ is increasing on $0\leq t\leq t_N$. By the monotonicity of $\Phi'$, 
\begin{equation}
   \begin{split}
       F_i^{(N)}(t_N)\leq F_i^{(N)}(0)<B_N,
   \end{split}
\end{equation}
which contradicts our assumption that $F_i^{(N)}(t_N)=B_N$. Therefore for all $0\leq t\leq T$,
\begin{equation}
    \begin{split}
        &\left|\dfrac{d}{dt}u_i^{(N)}(t)\right|
        =\left|A_i^{(N)}(t)\right|\leq B_N,\\
        &\left|u_i^{(N)}(t)\right|
        =
        \left|u_i(0)+\int_0^t A_i^{(N)}(s)ds\right|
        \leq
        U+B_NT,\\
        &\left|x_i^{(N)}(t)\right|
        =
        \left|x_i(0)+\int_0^t u_i^{(N)}(s)ds\right|
        \leq
        X+UT+B_NT^2,
    \end{split}
\end{equation}
and $\left|x_i^{(N)}(t)-x_i^{(N)}(t)\right|$ increases on $[0,T]$. 
\end{proof}

Note that boundedness of positions implies
\begin{equation} \label{BolzanoWeierstrass}
\begin{split}
    \inf_{\substack{i,j \in \N\\ i\neq j}} 
    \left |x_i(0) - x_j(0)\right | = 0.
\end{split}
\end{equation}
Then \eqref{B_NCondition} gives
\begin{equation} \label{someB_N}
\begin{split}
     B_N^2 
     \leq 
     \dfrac
     {\displaystyle \min_{1\leq i \neq j \leq N}X_{ij}\cdot U_{ij}}
     {2T^3}
     \leq
      2U
      \dfrac
     {\displaystyle \min_{1\leq i \neq j \leq N}\left|X_{ij}\right|}
     {2T^3}\to 0.
\end{split}
\end{equation}

\begin{Exam}
The alignment condition \eqref{align} is easily satisfied as the following examples show: 
\begin{enumerate}
\item  $u_i(0) = \lambda x_i(0)$, $0 < \lambda $ (bursts).
\item
All $x_i(0)$'s are on the $(x,y)$ plane and all $u_i(0)$'s are of the form $u_i(0) = (\alpha, \beta, \pm \gamma)$ for some fixed $\alpha, \beta, \gamma$.
\item
All $x_i(0)$'s are of the form $x_i(0) = (a_i, b_i, 0)$ and the $u_i(0)$'s are $u_i(0) = (a_i,b_i, c_i)$.
\end{enumerate}
\end{Exam}

\begin{Exam}
For $-\Phi'(r) = r^{-p}$, $p>1$, \eqref{sigma_N} becomes
\begin{equation}
\begin{split}
    \sigma_N^{p-1}
    \sum_{\substack{j = 1\\ j \neq i}}^{N}
    \frac{1}{\left|X_{ij}\right|^p}
    \leq B_N.
\end{split}
\end{equation}
Then it is enough to set 
\begin{equation}\label{sigma_N fixed}
\begin{split}
     {\sigma_N}
     = 
     \left(\frac{B_N}{N}\right)^{1/(p-1)} 
     \min_{\substack{1\leq i,j\leq N\\ i\neq j}}
     {\left|X_{ij}\right|^{p/(p-1)}}.
\end{split}
\end{equation}
In addition, in the ``burst'' case, i.e. $u_i(0)=\lambda x_i(0)$, \eqref{B_NCondition} is satisfied if we set \begin{equation}
\begin{split}
    B_N
    =
    \frac
     {\displaystyle \min_{1\leq i \neq j \leq N}X_{ij}\cdot U_{ij}}
     {4XT + 6UT^2 + 2T^3}
     =
     \frac
     {\displaystyle \min_{1\leq i \neq j \leq N}\lambda \left|X_{ij}\right|^2}
     {4XT + 6UT^2 + 2T^3}.
\end{split}
\end{equation}
Since we are adding molecules in a bounded domain, a reasonable case is\begin{equation}
   \begin{split}
       \displaystyle \min_{1\leq i \neq j \leq N} \left|X_{ij}\right|
       =
       \alpha N^{-1/3}.
   \end{split}
\end{equation}
With this choice, 
\begin{equation}
   \begin{split}
       B_N
       =
        \frac
     {\lambda\alpha^2 }
     {4XT + 6UT^2 + 2T^3}
      N^{-2/3}
     =
     \beta N^{-2/3}
   \end{split}
\end{equation}
Then \eqref{sigma_N fixed} becomes
\begin{equation}
   \begin{split}
       \sigma_N
       =
       \beta^{1/(p-1)}\alpha ^{p/(p-1)}
     N^{-(5+p)/3(p-1)} \to \alpha N^{-1/3}, \quad
     p\to \infty.
   \end{split}
\end{equation}
\end{Exam}

\begin{Cor}
There are $\sigma_N$'s such that the uniform estimates are always satisfied for Maxwellian $-\Phi'(r) = r^{-p}$, $p>1$ and bounded and aligned initial conditions.  
\end{Cor}

Under the same initial conditions and choice of $\sigma_N$ as in Theorem \ref{PicardThm}, and after iterating the proof of the Theorem, the standard Picard iteration converges for all $t$ in $[0,T]$. Details appeared in the second author's thesis \cite{J}.

\begin{Thm}
Assume $\Phi$, $\{x_i(0), u_i(0)\}_{i \in \N}$, and $\sigma_N$ as in the statement of Theorem \ref{PicardThm}. Then all conditions of Theorem \ref{Main} are satisfied, and $\frak{I}_{\Phi}$ now vanishes. In particular, the momentum equation is
\begin{equation} \label{w/oPhi}
\begin{split}
    & \int_0^T\int_{\R^3}
   \frac{\partial \varphi}{\partial t}
    (t, x)  u (t, x)  \mu_t(d x) dt
   +
   \int_0^T\int_{\R^3}
   \nabla \varphi(t, x)\cdot  u(t, x)   u(t, x) \mu_t (d x)  dt 
    \\
    &
    =
    -
    \int_0^T\int_{\R^3}
    \nabla \varphi(t,x) 
    \cdot\int ( v - u(t, x) )
    ( v -  u(t, x))\, M_{t, x}(dv)\ \mu_t (d x) dt
    .
\end{split}
\end{equation}
\end{Thm}

\begin{proof}
That bounded velocities and accelerations imply that the condition of Theorem \ref{Main} are satisfied was established in section \ref{bounded=>assumptionssatisfied}. 
Therefore the continuity and momentum equations are both satisfied. 

For the $\frak I_{\Phi}$ term first note that as $\displaystyle \left| \frac{d}{dt}u_i^{(N)}(t)\right|\leq B_N$,
\begin{equation}
\begin{split}
    \left |
    I_{\Phi}^{(N)}(t,\phi)
    \right|
     &=
     \left|
     -\frac1{N}
    \sum_i 
    \phi(t, x_i^{(N)}(t))
     \frac{d}{dt}u_i^{(N)}(t)
    \right|\\
%    &\leq
%    -\frac1{N}
%    \sum_i 
%    |\phi(t, x_i^{(N)}(t))|
%    \frac{1}{\sigma_N}
%    \sum_{j,\ j \neq i} 
%    \Phi'\left(\frac{| x_i^{(N)} (0)-  x_j^{(N)}(0)|}{\sigma_N}\right)
%\\ 
    &\leq
    C_{\phi}B_N.
\end{split}
\end{equation}
Therefore, after using \eqref{someB_N}, obtain $\left | I_{\Phi}^{(N)}(t,\phi)\right |\to 0$, as $ N \to \infty$, for any $\phi$.
\end{proof}

Equation \eqref{w/oPhi} is identical to Maxwell's equation (76) in \cite{Max}, with $u$ and all averages now rigorously defined, when there are no external forces and all molecules are identical. 

Maxwell obtains his equation assuming: i) elastic, binary collisions, ii) $-\Phi'(r) = r^{-5}$, iii) molecular chaos, and iv) negligible interaction between molecules ``not in the same volume element," see \cite{Max} p. 70.
We have deduced the same equation here rigorously assuming  Hamiltonian dynamics. Maxwell's assumption on negligible interactions is here reflected by rescaling $\Phi$ by $\sigma_N$ satisfying 
\eqref{sigma_N}. 

We recall here that Maxwell also argues on how to approximate the hydrodynamic equation \eqref{w/oPhi} up to certain order so that it becomes the compressible Navier-Stokes in the case of laminar flows. In particular, \eqref{w/oPhi} contains information on the transport coefficients of the macroscopic system. This will be presented elsewhere.
}

\end{document}